\documentclass[runningheads]{llncs}
\pdfoutput=1
\usepackage[hyphens]{url}
\usepackage{graphicx}
\usepackage{subfig}
\usepackage{color}
\usepackage{xcolor}
\usepackage[normalem]{ulem} 
\usepackage{amsmath,amssymb,amsfonts}
\usepackage{cite}
\usepackage{subfiles}
\graphicspath{{./figures/}}

\newcommand{\delete}[1]{}

\begin{document}
\title{A Distributed Simplex Architecture for Multi-Agent Systems}
\author{Usama Mehmood\inst{1} \and
Scott D. Stoller\inst{1} \and
Radu Grosu\inst{2} \and\\
Shouvik Roy\inst{1} \and
Amol Damare\inst{1} \and
Scott A. Smolka\inst{1}}
\authorrunning{U. Mehmood et al.}
\institute{Department of Computer Science, Stony Brook University, USA \and
Department of Computer Engineering, Technische Universit\"at Wien, Austria }

\maketitle             

\begin{abstract}
We present \emph{Distributed Simplex Architecture}~(DSA), a new runtime assurance technique that provides safety guarantees for multi-agent systems (MASs). DSA is inspired by the Simplex control architecture of Sha et al., but with some significant differences. The traditional Simplex approach is limited to single-agent systems or a MAS with a centralized control scheme. DSA addresses this limitation by extending the scope of Simplex to include MASs under distributed control. 
In DSA, each agent has a local instance of traditional Simplex such that the preservation of safety in the local instances implies safety for the entire MAS. 
We provide a proof of safety for DSA, and present experimental results for several case studies, including flocking with collision avoidance, safe navigation of ground rovers through way-points, and the safe operation of a microgrid.
 
\end{abstract}

\section{Introduction}      
\label{sec:intro}
A multi-agent system (MAS) is a group of autonomous, intelligent agents that work together to solve tasks and carry out missions.
MAS applications include the design of power systems and smart-grids \cite{Nasir2019_microgrids, Boussaada2016_smartgrids}, autonomous control of robotic swarms for monitoring, disaster management, military battle systems, etc. \cite{tahir2019}, and sensor networks. Many MAS applications are safety-critical.
It is therefore paramount that MAS control strategies ensure safety. 

In this paper, we present the \emph{Distributed Simplex Architecture} (DSA), a new runtime assurance technique that provides safety guarantees for MASs under distributed control. 
DSA is inspired by Sha et al.'s Simplex Architecture \cite{seto1999, sha2001}, but differs from it in significant aspects. 
The Simplex Architecture provides runtime assurance of safety by switching control from an unverified (hence potentially unsafe) \emph{advanced controller} (AC) to a verified-safe \emph{baseline controller} (BC), if the action produced by the AC could result in a safety violation in the near future. 
The switching logic is implemented in a verified \emph{decision module} (DM).
The applicability of the traditional Simplex Architecture is limited to systems with a centralized control architecture. 

DSA, illustrated in Fig. \ref{fig:DSA_arch}, addresses this limitation by making necessary additions to the traditional Simplex to widen its scope to include MASs. Also, as in \cite{dung2019_nsa}, it implements \emph{reverse switching}  by reverting control back to the AC when it is safe to do so.

\begin{figure}[t]
    \centering
    \includegraphics[width=12cm]{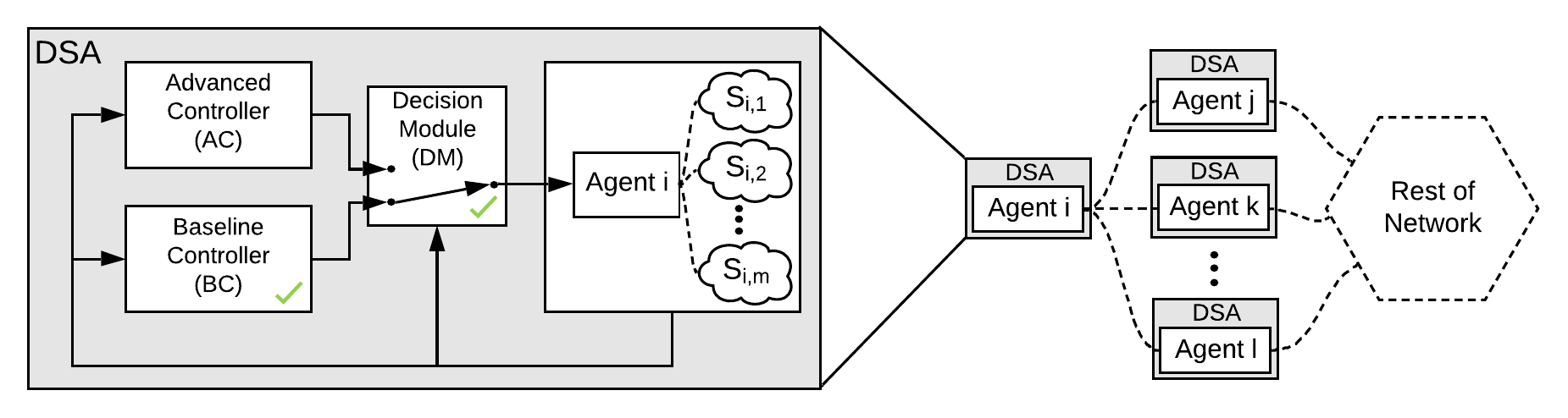}
    \caption{The DSA for the MAS on the right. All agents in the MAS are homogeneous and operate under DSA, but the figure shows the DSA components for only agent $i$. The sensed state of agent $i$'s $j^{th}$ neighbor is denoted as $S_{i,j}$. The AC,  BC, and  DM take as input the state of the agent and its neighbors.}
    \label{fig:DSA_arch}
    \vspace*{-4ex}
\end{figure}

In DSA, for each agent, there is a verified-safe BC and a certified switching logic such that if all the agents operate under DSA, then safety of the MAS is guaranteed.
The BC and DM along with the AC are distributed and depend only on local information. DSA itself is \emph{distributed} in the sense that it involves one local instance of traditional Simplex per agent such that the conjunction of their respective safety properties yields the desired safety property for the entire MAS.  For example, consider our flocking case study, where we want to establish collision-freedom for the entire MAS.  This can be accomplished in a distributed manner by showing that each local instance of Simplex, say for agent $i$, ensures collision-freedom for agent $i$ and its neighboring agents. 


DSA allows agents to switch their mode of operation independently. At any given time, some agents may be operating in AC mode while others are operating in BC mode. 
Our approach to the design of the BC and DM leverages \emph{Control Barrier Functions} (CBFs), which have been used to synthesize safe controllers \cite{Gurriet2018, egerstedt2018, magnus_heterogenous2016}, and are closely related to Barrier Certificates used for safety verification of closed dynamical systems \cite{prajna2004, prajna2006}.
A CBF is a mapping from the state space to a real number, with its zero level-set partitioning the state space into safe and unsafe regions.  If certain inequalities on the Lie derivative of the CBF are satisfied, then the corresponding control actions are considered safe (admissible).

In DSA, the BC is designed as an optimal controller with the goal of increasing a utility function based on the Lie derivatives of the CBFs. As CBFs are a measure of the safety of a state, optimizing for control actions with a higher Lie derivative values gives a direct way to make the state safer. The safety of the BC is further guaranteed by constraining the control action to remain in a set of admissible actions that satisfy certain inequalities on the Lie derivatives of the CBFs.
CBFs are also used in the design of the switching logic, as they provide an efficient method for 
checking whether an action could lead to a safety violation during the next time step.

We demonstrate the effectiveness of DSA on several example MASs, including a flock of robots moving coherently while avoiding inter-agent collisions, ground rovers safely navigating through a series of way-points, and safe operation of a microgrid.

\section{Background} 
\label{sec:background}
\subsection{Simplex Architecture}
The Simplex Control Architecture relies on a verified-safe baseline controller (BC) in conjunction with the verified switching logic of the Decision Module (DM) to guarantee the safety of the plant (Agent $i$ in the Fig.~\ref{fig:DSA_arch}) while permitting the use of an unverifiable, high-performance advanced controller (AC).

Let the \emph{admissible} states be those which satisfy all safety constraints and operational limits. Other states are are called \emph{inadmissible}. The goal of the Simplex Architecture is to ensure the system never enters an inadmissible state. The set $\mathcal{R}$ of \emph{recoverable} states is a subset of the admissible states such that the BC, starting from any state in $\mathcal{R}$ guarantees that all future states are also in $\mathcal{R}$. The recoverable set takes into account the inertia of the physical system, giving the BC enough time to preserve safety.


The DM's \emph{forward switching condition} (FSC) evaluates the control action proposed by the AC and decides whether to switch to the BC. 
A common technique to develop a FSC is to shrink the recoverable region by a margin based on the maximum time derivative of the state and the length of a timestep, and switch to BC if the current state lies outside this smaller set.

\subsection{Control Barrier Functions}
Control Barrier Functions (CBFs)~\cite{wieland2007,Ames2019_reviewCBF} are an extension of the Barrier Certificates used for safety verification of hybrid systems \cite{prajna2004, prajna2006}.  CBFs are a class of Lyapunov-like functions used to guarantee safety for nonlinear control systems by assisting in the design of a class of safe controllers that establish the forward-invariance of safe sets~\cite{magnus2015, magnus_heterogenous2016}.  Our presentation of CBFs is based on \cite{Ames2019_reviewCBF}.

Consider a nonlinear affine control system
\begin{equation}
\label{eq:control_affine}
    \dot{x} = f(x) + g(x)u,
\end{equation}
with state $x \in D \subset \mathbb{R}^n$, control input $u \in U$, and functions $f$ and $g$ that are locally Lipschitz. The set $\mathcal{R}$ of recoverable states is defined as the super-level set of a continuously differentiable function $h:D \subset \mathbb{R}^n \to \mathbb{R}$. The recoverable set $\mathcal{R}$ and its boundary $\delta \mathcal{R}$ are given by:
\begin{align}
\label{eq:recoverable_set}
    \mathcal{R} &= \{x \in D \subset \mathbb{R}^n | h(x) \geq 0 \} \\
    \mathcal{\delta R} &= \{x \in D \subset \mathbb{R}^n | h(x) = 0 \} 
\end{align}
For all $x \in \mathcal{D}$, if there exists an extended class $\mathcal K$ function $\alpha: \mathbb{R} \to \mathbb{R}$ (strictly increasing and $\alpha(0) = 0$) such that the following condition on the Lie-derivative of $h$ is satisfied:
\begin{equation}
\label{eq:cbf_define}
     \underset{u \in U}{sup}[L_fh(x) + L_gh(x)u + \alpha(h(x) \geq 0
\end{equation}
then the function $h(x)$ is a valid CBF.
Condition~(\ref{eq:cbf_define}) implies the existence of a control action for all $x \in D$, such that the Lie-derivative of $h$ is bounded from below by $-\alpha(h(x))$. Furthermore, for $x \in \delta \mathcal{R}$, condition~(\ref{eq:cbf_define}) reduces to a result for set invariance known as Nagumo's theorem \cite{nagumos_theorom_1, nagumos_theorom_2}. Condition~(\ref{eq:cbf_define}) is used to define the set $ K(x)$ of control actions that establish the forward invariance of set $\mathcal{R}$; i.e., starting from $x \in \mathcal{R}$, the state will always remain inside $\mathcal{R}$:
\begin{equation}
\label{eq:set_safe_acc_background}
    K(x) = \{u \in U : L_fh(x) + L_gh(x)u + \alpha(h(x)]\geq 0\}
\end{equation}
\begin{theorem}
\label{thm:backgrounf_cbf}
\cite{Ames2019_reviewCBF} For the control system given in Eq.~(\ref{eq:control_affine}) and recoverable set $\mathcal{R}$ defined in (\ref{eq:recoverable_set}) as the super-level set of some continuously differentiable function $h:\mathbb{R}^n \to \mathbb{R}$, if $h$ is a control barrier function for all $x \in D$ and $\frac{\delta h}{\delta x} \neq 0$ for all $x \in \delta \mathcal{R}$, then any controller $u$ such that $\forall x \in D: u(x) \in  K(x)$ 
ensures forward-invariance of $\mathcal{R}$.
\end{theorem}

\begin{proof}
Condition~(\ref{eq:cbf_define}) on the Lie-derivative of $h$ reduces, on the boundary of $\mathcal{R}$, to the set invariance condition of Nagumo's theorem: for ${x \in \delta \mathcal{R}}$, ${\dot{h} \geq - \alpha (h(x)) = 0}$. Hence, according to Nagumo's theorem \cite{nagumos_theorom_1, nagumos_theorom_2} the set $\mathcal{R}$ is forward-invariant.
\end{proof}

\section{Distributed Simplex Architecture}
\label{sec:DSA}
This section describes the Distributed Simplex Architecture (DSA). We formally introduce the MAS safety problem and then discuss the main components of DSA, namely, the distributed baseline controller (BC) and the distributed decision module (DM). 

We say that an instance of DSA is \emph{symmetric} if every agent uses the same switching condition and baseline controller. Moreover, DSA, or more precisely the MAS it is controlling, is \emph{homogeneous} if every constituent agent is an instance of the same plant model.

Consider a MAS consisting of $k$ homogeneous agents, denoted as $\mathcal{M} = \{1,...,k\}$, where the nonlinear control affine dynamics for the $i^{th}$ agent are:
\begin{equation}
\label{eq:control_affine_MAS}
    \dot{x}_i = f(x_i) + g(x)u_i,
\end{equation}
Here, $x_i \in D \in \mathbb{R}^n$ is the state of agent $i$ and $u_i \in U \subset \mathbb{R}^m$ is its control input.  
For an agent $i$, we define the set of its \emph{neighbors} $\mathcal{N}_i \subseteq \mathcal{M}$ as the agents whose state is accessible to the agent $i$ either through sensing or communication.
Depending on the application, the set of neighbors could be fixed or vary dynamically. 
For example, in our flocking case study (Section~\ref{sec:flocking_case_study}), agent $i$'s neighbors (in a given state) are the agents within a fixed distance $r$ of agent $i$; we assume agent $i$ can accurately sense the positions and velocities those agents.
 We denote the combined state of all the agents in the MAS as the vector ${\mathbf{x} = \{x_1^T, x_2^T, ...x_k^T\}^T}$ and denote the state of the neighbors of agent $i$ (including agent $i$) as $x_{\mathcal{N}_i}$.
DSA uses discrete-time control: the DM and controllers are evaluated every $\eta$ seconds.  
We assume 
that all agents evaluate DM and controllers simultaneously; this assumption simplifies the analysis but can be relaxed.

\vspace{-2.0ex}
\subsubsection{Admissible States}
Set of admissible states $\mathcal{A} \subset \mathbb{R}^{kn}$ consists of all states that satisfy the safety constraints. A constraint $C$ is a  function from $k$-agent MAS states to the reals; ${C:D^k \to \mathbb{R}}$. In this paper, we are primarily concerned with {binary constraints} (between neighboring agents) $C_{ij}: D \times D \to \mathbb{R}$, and \emph{unary constraints} $C_i:D \to \mathbb{R}$. Hence, set of admissible states, $\mathcal{A} \subset \mathbb{R}^{kn}$ are the states of MAS $\mathbf{x} \in \mathbb{R}^{kn}$, such that all the unary and binary constraints are satisfied.  

Formally, a symmetric instance of DSA aims to solve the following problem.  Given a MAS defined as in Eq.~(\ref{eq:control_affine}) and $\mathbf{x}(0) \in \mathcal{A}$, design a BC and DM to be used by all agents such that the MAS remains safe; i.e. $\mathbf{x}(t) \in \mathcal{A}$, $\forall \ t > 0$. 


\vspace{-2.0ex}
\subsubsection{Recoverable States}
For each agent $i$, the local admissible set $\mathcal{A}_i \subset \mathbb{R}^n$ is the set of states $x_i \in \mathbb{R}^n$ which satisfy all the unary constraints. The set $\mathcal{S}_i \subset \mathcal{A}_i$ is defined as the super-level set of the CBF $h_i:\mathbb{R}^n \to \mathbb{R}$, which is designed to ensure forward-invariance of $\mathcal{A}_i$.
Similarly, for a pair of neighboring agents $i,j$ where $i \in \mathcal{M}, j \in \mathcal{N}_i$, the pairwise admissible set $\mathcal{A}_{ij} \subset \mathbb{R}^{2n}$ is the set of pairs of states which satisfy all the binary constraints. The set $\mathcal{S}_{ij} \subset \mathcal{A}_{ij}$ is defined as the super-level set of the CBF $h_{ij}:\mathbb{R}^{2n} \to \mathbb{R}$ designed to ensure forward-invariance of $\mathcal{A}_{ij}$.
The recoverable set $\mathcal{R}_{ij} \subset \mathbb{R}^{2n}$, for a pair of neighboring agents $i, j$ where ${i \in \mathcal{M}, j \in \mathcal{N}_i}$, is defined in terms of $\mathcal{S}_{i}$, $\mathcal{S}_{j}$ and $\mathcal{S}_{ij}$.
\begin{align}
\label{eq:recoverable_sets}
    \mathcal{S}_i &= \{x_i \in \mathbb{R}^n | h_i(x_i) \geq 0 \} \\
    \mathcal{S}_{ij} &= \{(x_i, x_j) \in \mathbb{R}^{2n} | h_{ij}(x_i, x_j) \geq 0 \}  \\
    \mathcal{R}_{ij} &= (\mathcal{S}_i \times \mathcal{S}_j) \cap \mathcal{S}_{ij} 
\end{align}
The recoverable set $\mathcal{R} \subset \mathcal{A}$ for the entire MAS is defined as the set of system states in which $(x_i,x_j) \in \mathcal{R}_{ij}$ for every pair of neighboring agents $i,j$.
The CBFs can be computed 
using sum-of-squares programming \cite{cbf_sos} or the technique in \cite{cbf_synthesis}. An application of these techniques for the synthesis of CBFs for several systems can be found in \cite{Ames2019_reviewCBF}. 
Note that if agent $i$ and $j$'s controllers satisfy the following constraints based on the Lie derivatives of $h_i, h_j$ and $h_{ij}$, similar to the constraints in (\ref{eq:set_safe_acc_background}), the pairwise state of agents $i$ and $j$ will remain in $\mathcal{R}_{ij}$ according to Theorem \ref{thm:backgrounf_cbf}. \begin{subequations}
\label{eq:lie_constraints_dsa}
    \begin{equation}
        L_fh_i(x_i) + L_gh_i(x_i) u_i + \alpha(h_i(x_i)) \geq 0 
    \end{equation}
    \begin{equation}
        L_fh_j(x_j) + L_gh_j(x_j) u_j + \alpha(h_j(x_j)) \geq 0 \\
    \end{equation}
    \begin{equation}
    \label{eq:lie_constraints_dsa_c}
        L_fh_{ij}(x_i, x_j) + L_gh_{ij}(x_i, x_j) \begin{bmatrix}u_i\\u_j\end{bmatrix} + \alpha(h_{ij}(x_i, x_j)) \geq 0 
    \end{equation}
\end{subequations}

\vspace{-2.0ex}
\subsubsection{Constraint Partitioning}
Note that the constraints in (\ref{eq:lie_constraints_dsa}) are linear in the control variable. For ease of notation we write the unary constraints as ${A_i u_i \leq b_i}$ and the binary constraints as ${\left[\begin{smallmatrix}P_{ij}&Q_{ij}\end{smallmatrix} \right] \left[\begin{smallmatrix}u_i\\u_j\end{smallmatrix} \right] \leq b_{ij}}$.

The binary constraint in (\ref{eq:lie_constraints_dsa_c}) is a condition on the control action of a pair of agents. For a centralized MAS, the global controller can pick coordinated actions for agents $i$ and $j$ to ensure the binary constraint (\ref{eq:lie_constraints_dsa_c}) is satisfied.  However, for a decentralized MAS, the distributed control of the two agents cannot independently satisfy the binary constraint without running an agreement protocol.

As DSA is a distributed control framework, we solve the problem of the satisfaction of the binary constraint by partitioning the binary constraint into two unary constraints such that the satisfaction of the unary constraints implies the satisfaction of the binary constraint (but not vice versa)~\cite{magnus_heterogenous2016}. 
\begin{equation}
\label{eq:constraint_partition}
    \begin{bmatrix}P_{ij} & Q_{ij}\end{bmatrix}
    \begin{bmatrix}u_i \\ u_j\end{bmatrix} \leq b_{ij} \;\to\;
    \begin{cases} 
    P_{ij} u_i \leq b_{ij}/2 \\
    Q_{ij} u_j \leq b_{ij}/2
    \end{cases}
\end{equation}
The satisfaction of the two unary constraints in (\ref{eq:constraint_partition}) by the respective controllers of agents $i$ and $j$ guarantees safety because the binary constraint still holds. Moreover, the equal partitioning ensures that the agents share an equal responsibility to keep the pairwise state safe.
The admissible control space for an agent i, denoted by $\mathcal{L}_i$, is an intersection of half-spaces of the hyper-planes defined by the linear constraints. 
\begin{equation}
\label{eq:admissible_control_space}
    \mathcal{L}_i = \{u_i \in U \mid \forall j \in \mathcal{N}_i: A_i u_i \leq b_i \land P_{ij} u_i \leq b_{ij}  \}
\end{equation}

\begin{theorem}
Given a MAS indexed by $\mathcal{M}$ and with dynamics in (\ref{eq:control_affine_MAS}), if the controller for each agent $i \in \mathcal{M}$ chooses an action $u_i \in \mathcal{L}_i$, thereby satisfying the Lie-derivative constraints on the respective CBFs, and $\mathbf{x}(0) \in \mathcal{R}$, then the MAS is guaranteed to be safe. 
\end{theorem}
\begin{proof}
If all the agents choose an action from their respective admissible control spaces $\mathcal{L}_i$, then the forward invariance of the set $\mathcal{S}_i$ for all $i \in \mathcal{M}$ and  
$\mathcal{S}_{ij}$ for all $i \in \mathcal{M}$, $j \in \mathcal{N}_i$ is established by Theorem~\ref{thm:backgrounf_cbf}. Therefore, $\mathcal{R}_{ij}$ is forward invariant for all $i \in \mathcal{M}, j \in \mathcal{N}_i$ and consequently $\mathcal{R}$ is forward invariant.
\end{proof}

\vspace{-2.0ex}
\subsection{Baseline Controller}
The BC is a distributed controller with the task to keep the state of the agent in the safe region. For an agent $i$, the control law of the BC depends on its state $x_i$ and the states of its .
In our design, the BC considers only the safety-critical aspects, leaving the mission-critical objectives to the AC.  Specifically, the BC is designed to move the system toward safer states as quickly as possible. This reduces the width of the necessary ``safety margin'' between unsafe and recoverable states, allowing a looser FSC (i.e., allowing the AC to stay in control more often).

We design the BC as the solution to the following constrained multi-objective optimization (MOO) problem where the utility function is a weighted sum of objective functions based on the Lie derivatives of the CBFs $h_i$ and $h_{ij}$ introduced above:
\begin{equation}
\label{eq:bc}
\begin{aligned}
    u_i^* = \underset{u_i}{argmax} \quad & \frac{1}{h_i} (L_f h_{i}+L_g h_{i}u_i)  +\sum_{j \in \mathcal{N}_i} \frac{1}{h_{ij}} (L_f h_{ij}+L_g h_{ij}\begin{bmatrix}u_i\\0\end{bmatrix}) \\
    \textrm{s.t.} \quad & u_i \in \mathcal{L}_i
\end{aligned}
\end{equation}
The bottom component of the column vector in the last term is agent $i$'s prediction for agent $j$'s next control action $u_j$.  Since we consider MASs in which agents are unable to communicate planned control actions, agent $i$ simply predicts that $u_j=0$.  This approach has been shown to work well in prior work on distributed model-predictive control for flocking \cite{Mehmood2018}.

Recall that, by definition, the CBFs quantify the degree of safety of the state with respect to the given  safety constraints, with larger (positive) values indicating safer states. For a given state, the Lie derivative of a CBF is a linear function in the control action. A positive value of the Lie derivative indicates that the proposed action will lead to a state which has a higher CBF value and therefore is safer. 

The solution to the optimization problem in (\ref{eq:bc}) is a control action that maximizes the weighted sum of the Lie derivatives of the CBFs.
We note that in a weighted-sum formulation of a MOO problem, it is possible that some objective functions are negative in the optimal solution.  We ensure the selected action $u_i$ is safe by constraining $u_i$ to be in the admissible control space $\mathcal{L}_i$, defined in (\ref{eq:admissible_control_space}).

The weights in the utility function in (\ref{eq:bc}) prioritize certain safety constraints over others.  We use state-dependent weights which are the inverses of the CBF functions, thereby giving more weight to maximizing the Lie derivatives of CBFs corresponding to safety constraints that are closer to being violated.

\vspace{-2.0ex}
\subsection{Decision Module}

Each agent's DM implements the switching logic for both forward switching and reverse switching. Control is switched from the AC to the BC if the \emph{forward switching condition} (FSC) is true. Similarly, control is reverted back to the AC (from the BC) if the \emph{reverse switching condition} (RSC) is true. For an agent $i$, the state of the DM is denoted as $DM_i \in \{AC, BC\}$, with $DM_i = AC$ ($DM_i = BC$) indicating that the advanced (baseline) controller is in control. DSA starts with all agents in the AC mode; i.e., $DM_i(t) = AC$ for all $t \leq 0$ and $i\in \mathcal{M}$; this is justified by the assumption that $\mathbf{x}(0)\in \mathcal{R}$. For $t > 0$, the DM state is given by:
\begin{equation}
\label{eq:dm}
    DM_i(t)= 
\begin{cases}
    AC          & \text{if } DM_i(t-1)=BC \text{ and } RSC(x_{\mathcal{N}_i})\\
    BC          & \text{if } DM_i(t-1)=AC \text{ and } FSC(x_{\mathcal{N}_i})\\
    DM_i(t-1)   & \text{otherwise}
\end{cases}
\end{equation}
where $x_{\mathcal{N}_i}$
is the vector containing the states of agent $i$ and its neighbors. 

We derive the switching conditions from the CBFs as follows. 
To ensure safety, the FSC must be true in a state $x_{\mathcal{N}_i}(t)$ if an unrecoverable state is reachable from $x_{\mathcal{N}_i}(t)$ in one time step $\eta$.
For a CBF function, in a given state, we define a \emph{worst-case} action to be an action that minimizes the Lie derivative of the CBF. The check for one-step reachability of an unrecoverable state is based on the minimum value of the Lie derivative of the CBFs, which corresponds to the worst-case actions by the agents. Hence, for each CBF $h$, we define a minimum threshold value $\lambda_h(x_{\mathcal{N}_i})$ equal to the magnitude of the minimum of the Lie derivative of the CBF times $\eta$, and we switch to BC if, in the current state, the value of any CBF $h$ is less than $\lambda_h(x_{\mathcal{N}_i})$.  This results in a FSC of the following form:
\begin{equation}
\label{eq:FSC_DSA}
FSC(x_{\mathcal{N}_i}) = (h_i < \lambda_{h_i}(x_{\mathcal{N}_i})) \vee  (\exists j \in \mathcal{N}_i \ | \ h_{ij} < \lambda_{h_{ij}}(x_{\mathcal{N}_i})) 
\end{equation}
Thus, the one-step reachability check shrinks the size of the recoverable set by an amount equal to the maximum change that can occur from the current state in one control period, and the switch occurs if the current state is outside this smaller set.

We derive the RSC using a similar approach, except based on an $m$-time-step reachability check with $m>1$, in order to prevent frequent switching between AC and BC.  The RSC holds if, in the current state, the value of each CBF $h$ is greater than the threshold $m\lambda_h(x_{\mathcal{N}_i})$.
\begin{equation}
\label{eq:RSC_DSA}
RSC(x_{\mathcal{N}_i}) = (h_i > m\lambda_{h_i}(x_{\mathcal{N}_i})) \wedge  (\forall j \in \mathcal{N}_i \ | \ h_{ij} > m\lambda_{h_{ij}}(x_{\mathcal{N}_i})) 
\end{equation}
This definition of the RSC ensures that, when control is switched to AC, the state is safe and the FSC will not hold for at least $m$ time steps.  


\vspace{-2.0ex}
\subsection{Safety Theorem}
Our main result is the following safety theorem for DSA.

\begin{theorem}
Given an MAS indexed by $\mathcal{M}$ with dynamics in (\ref{eq:control_affine_MAS}), if each agent operates under DSA with the BC as defined in (\ref{eq:bc}), the DM as defined in (\ref{eq:dm}), and $\mathbf{x}(0) \in \mathcal{R}$, then the MAS will remain safe.
\end{theorem}

\begin{proof}
The proof proceeds by considering both DM states for an arbitrary agent $i$ and establishing that its next state is safe. 
First, consider an agent $i$ at time $t$ with $DM_i(t)=AC$. As the FSC is false, the one-step reachability check in the FSC ensures that the CBFs for unary and binary safety constraints are strictly positive at the next state $x_i(t+\eta)$, i.e. $h_i(x_i(t+\eta)) > 0$ and $\forall j \in \mathcal{N}_i: h_{ij}(x_{i}(t+\eta), x_{j}(t+\eta)) > 0$, hence the next state is recoverable.
Next, consider an agent $i$ at time $t$ with $DM_i(t)=BC$. We divide the neighbors of $i$ into two sets based on their DM states: the sets of neighbors in AC mode and BC mode are denoted as $\mathcal{N}_i^{AC}$ and $\mathcal{N}_i^{BC}$, respectively.
The agents in BC mode choose their control actions from their corresponding admissible control spaces as defined in Eq.~\ref{eq:admissible_control_space}. Hence, according to \emph{Theorem}~2, these agents will satisfy unary safety constraints and pairwise safety constraints among themselves. 
As for the neighbors in AC mode, due to the one-step reachability check in their FSC, in the state $x_i(t+\eta)$, the pairwise CBFs satisfy $h_{ij}(x_i(t+\eta), x_j(t+\eta)) \geq 0$ for all $j \in \mathcal{N}_i^{AC}$. Hence, $x_i(t+\eta)$ is recoverable for $DM_i(t) = BC$.
We have proven that, for any agent $i$ and time step $t$, if $x_i(t)$ is recoverable, then $x_i(t+\eta)$ is recoverable. By assumption, $\mathbf{x}(0) \in \mathcal{R}$.  Therefore, by induction, $\mathbf{x}(t) \in \mathcal{R}$ for $t > 0$.
\end{proof}

\section{Flocking Case Study}
\label{sec:flocking_case_study}
We evaluate DSA on the distributed flocking problem with the goal of preventing inter-agent collisions. 
Consider an MAS consisting of $n$ robotic agents, indexed by $\mathcal{M} = \{1,\ldots,n\}$ with double integrator dynamics:
\begin{equation}
\label{eq:flocking_dynamics_continuous}
\begin{bmatrix}
\dot{p}_i\\
\dot{v}_i 
\end{bmatrix} 
=
\begin{bmatrix}
0 & I_{2 \times 2}\\ 
0 & 0 
\end{bmatrix}
\begin{bmatrix}
p_i\\
v_i 
\end{bmatrix}
+
\begin{bmatrix}
0\\
I_{2 \times 2} 
\end{bmatrix}
a_i
\end{equation}
where $p_i$, $v_i$, $a_i$ $\in \mathbb{R}^2$ are the position, velocity and acceleration of agent $i \in \mathcal{M}$, respectively.
The magnitudes of velocities and accelerations are bounded by $\bar{v}$ and $\bar{a}$, respectively.
Acceleration $a_i$ is the control input for agent $i$.
As DSA is a discrete-time protocol, the state of the DM and the $a_i$'s are updated every $\eta$ seconds.
The \emph{state} of an agent $i$ is denoted by the vector $s_i = [p_i^T v_i^T]^T$.
The \emph{state} of the entire flock at time $t$, is denoted by the vector ${\mathbf{s}(t) = [\mathbf{p}(t)^T \ \mathbf{v}(t)^T]^T \in \mathbb{R}^{4n}}$, where $\mathbf{p}(t) = [p_1^T(t) \cdot \cdot \cdot p_n^T(t)]^T$ and $\mathbf{v}(t) = [v_1^T(t) \cdot \cdot \cdot v_n^T(t)]^T$ are the vectors respectively denoting the positions and velocities of the flock at time $t$.

We assume that an agent can accurately sense the positions and velocities of nearby agents within a fixed distance $r$.
The set of the $\emph{spatial neighbors}$ of agent $i$
is defined as ${\mathcal{N}_{i}(\mathbf{p}) = \left\lbrace j \in \mathcal{M} \mid j\ne i \land \| p_i - p_j \| < r\right\rbrace}$,
where $\| \cdot \|$ denotes the Euclidean norm.
For ease of notation, we sometimes use $\mathbf{s}$ and $\mathbf{s}_i$ to refer to the state variables $\mathbf{s}(t)$ and $\mathbf{s}_i(t)$, respectively, without the time index.

The MAS is characterized by a set of operational constraints which include physical limits and safety properties. States that satisfy the operational constraints are called \emph{admissible}, and are denoted by the set $\mathcal{A} \in \mathbb{R}^{4n}$.
The desired safety property is that no pair of agents is in a ``state of collision''.
A pair of agents is considered to be in a \emph{state of collision} if the Euclidean distance between them is less than a threshold distance $d_{min} \in \mathbb{R}^+$, resulting in a binary safety constraint of the form: ${\left \| p_i - p_j \right \| - d_{min} \geq 0} \ \forall \ i \in \mathcal{M}, j \in \mathcal{N}_i$. 
Similarly, a state $\mathbf{s}$ is $recoverable$ if all pairs of agents can brake (de-accelerate) relative to each other without colliding. Otherwise, the state $\mathbf{s}$ is considered $\emph{unrecoverable}$. 
\vspace{-2.0ex}
\subsection{Synthesis of Control Barrier Function}
Let $\mathcal{R}_{ij} \subset \mathbb{R}^{8}$ be the set of recoverable states for a pair of agents $i,j \in \mathcal{M}$. The flock-wide set of recoverable states, denoted by $\mathcal{R} \subset \mathbb{R}^{4n}$, is defined in terms of $\mathcal{R}_{ij}$.
As in \cite{magnus2015}, the set $\mathcal{R}_{ij}$ is defined as the super-level set of a pairwise CBF $h_{ij}:\mathbb{R}^{8} \to \mathbb{R}$: ${\mathcal{R}_{ij} = \left\lbrace s_i, s_j\mid h_{ij}(s_i, s_j) \geq 0 \right\rbrace}$.
The flock-wide set of recoverable states $\mathcal{R} \subset \mathcal{A}$ is defined as the set of system states in which $(s_i,s_j) \in \mathcal{R}_{ij}$ for every pair of neighboring agents $i,j$.

In accordance with~\cite{magnus2015}, the function $h_{ij}(s_i, s_j)$ is based on a safety constraint over a pair of agents $i,j \in \mathcal{M}$.
The safety constraint ensures that for any pair of agents, the maximum braking force can always keep the agents at a distance greater than $d_{min}$ from each other. As introduced earlier, $d_{min}$ is the threshold distance that defines a collision. Considering that the tangential component of the relative velocity, denoted by $\Delta v$, causes a collision, the constraint regulates $\Delta v$ by application of maximum acceleration to reduce $\Delta v$ to zero. Hence, the safety constraint can be represented as the following condition on the inter-agent distance ${\left \| \Delta \mathbf{p}_{ij}\right \| = \left \| p_i - p_j \right \|}$, the braking distance $(\Delta v)^2/4\bar{a}$, and the safety threshold distance $d_{min}$:
\begin{gather} 
\label{eq:safety_constraint}
        \left \| \Delta \textbf{p}_{ij} \right \| - \frac{(\Delta v)^2}{4\bar{a}} \geq d_{min} \\
\label{eq:cbf}
        h_{ij}(s_i, s_j) = \sqrt{4\bar{a}( \left \| \Delta \textbf{p}_{ij} \right \| - d_{min}} ) - \Delta v \geq 0
\end{gather}
The braking distance is the distance covered while the relative speed reduces from $\Delta v$ to zero under a deceleration of $2\bar{a}$. The constraint in Eq. (\ref{eq:safety_constraint}) is re-arranged to get the CBF $h_{ij}$ given in Eq. (\ref{eq:cbf}). 

Combining~(\ref{eq:cbf}) and~(\ref{eq:lie_constraints_dsa_c}), we arrive at the linear constraint on the accelerations for agents $i$ and $j$, which constrains the Lie derivative of the CBF in~(\ref{eq:cbf}) to be greater than $-\alpha(h_{ij})$. We set ${\alpha(h_{ij}) = \gamma h_{ij}^3}$, as in~\cite{magnus2015}, resulting in the following constraint on the accelerations of agents $i,j$:
\begin{equation}
\begin{split}
\label{eq:flocking_lie_der}
    \frac{\Delta \mathbf{p}_{ij}^T (\Delta \mathbf{a}_{ij}) }{\left \| \Delta \mathbf{p}_{ij} \right \|} 
- \frac{(\Delta \mathbf{v}_{ij}^T \Delta \mathbf{p}_{ij})^2}{\left \| \Delta \mathbf{p}_{ij} \right \|^3}
&+ \frac{\left \| \Delta \mathbf{v}_{ij} \right \|^2}{\left \| \Delta \mathbf{p}_{ij} \right \|} \\
&+ \frac{2\bar{a}\Delta \mathbf{v}_{ij}^T \Delta \mathbf{p}_{ij}}{\left \| \Delta \mathbf{p}_{ij} \right \|\sqrt{ 4\bar{a} (\left \| \Delta \mathbf{p}_{ij} \right \|- d_{min})}}
\geq
- \gamma h_{ij}^3
\end{split}
\end{equation}
where the left-hand side is the Lie derivative of the CBF $h_{ij}$ and ${\Delta \mathbf{p}_{ij} = p_i - p_j}$, ${\Delta \mathbf{v}_{ij} = v_i - v_j}$, and ${\Delta \mathbf{a}_{ij} = a_i - a_j}$ are the vectors representing the relative position, the relative velocity, and the relative acceleration of agents $i$ and $j$, respectively.
We further note that the binary constraint~(\ref{eq:flocking_lie_der}) can be represented as ${\left[\begin{smallmatrix}P_{ij}&Q_{ij}\end{smallmatrix} \right] \left[\begin{smallmatrix}a_i\\a_j\end{smallmatrix} \right] \leq b_{ij}}$, and hence it can be split into two unary constraints (${P_{ij} u_i \leq b_{ij}/2}$ and ${Q_{ij} u_j \leq b_{ij}/2}$), following the convention in (\ref{eq:constraint_partition}).
The set of safe accelerations for an agent $i \in \mathcal{A}$, denoted by $\mathcal{K}_i(\mathbf{s}_i) \subset \mathbb{R}^2$, is defined as the intersection of the half-planes defined by the Lie-derivative-based constraints, where each neighboring agent contributes a single constraint:
\begin{equation}
\label{eq:set_safe_acc}
    \mathcal{K}_i(\mathbf{s}_i) = \left\lbrace a_i \in \mathbb{R}^2 \mid P_{ij} u_i \leq b_{ij}/2, \ \forall j \in \mathcal{N}_i \right\rbrace
\end{equation}
With the CBF for collision-free flocking defined in~(\ref{eq:cbf}) and the admissible control space defined in~(\ref{eq:set_safe_acc}), the BC, RSC, and FSC follow from (\ref{eq:bc}), (\ref{eq:FSC_DSA}), and (\ref{eq:RSC_DSA}), respectively.
\vspace{-2.0ex}
\subsection{Advanced Controller}
We use the Reynolds flocking model~\cite{reynolds1987} as the AC. In the Reynolds model, the acceleration $a_i$ for each agent is a weighted sum of three acceleration terms, based on simple rules of interaction with the neighboring agents: \emph{separation} (move away from your close neighbors), \emph{cohesion} (move towards the centroid of your neighbors), and \emph{alignment} (match your velocity with the average velocity of your neighbors). 
The acceleration for agent $i$ is ${a_i = w_s a_i^{s} + w_c a_i^{c} + w_{al} a_i^{al}}$,
where ${w_s, w_c, w_{al} \in \mathbb{R}^+}$ are the scalar weights and ${a_i^s, a_i^c, a_i^{al} \in \mathbb{R}^2}$ are the acceleration terms corresponding to separation, cohesion, and alignment, respectively. 
We note that the Reynolds model does not guarantee collision avoidance.
Nevertheless, when the flock stabilizes, the average distance to the closest neighbors are determined by the choice of the weights of the interaction terms.
\vspace{-2.0ex}
\subsection{Experimental Results}
The number of agents in the MAS is $n = 15$. The other parameters used in the experiments are $r = 4$, $\bar{a} = 5$, $\bar{v} = 2.5$, $d_{min} = 2$, and $\eta = 0.1$s. The length of the simulations is 50 seconds. The initial positions and the initial velocities are uniformly sampled from $[-10,10]^2$ and $[-1,1]^2$, respectively, and we ensure that the initial state is recoverable.
The weights of the Reynolds' model terms are picked experimentally to ensure that no pair of agents are in a state of collision in the steady state. They are set to $w_s = 3$, $w_c = 1.5$, and $w_{al} = 0.5$. 
\begin{figure}[t]
\centering
\includegraphics[width=0.45\textwidth, trim=1.3cm 8.45cm 1cm .30cm, clip]{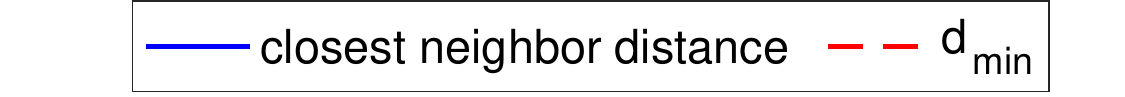}\\ \vspace*{-9pt}
\subfloat[Reynolds Model]{\includegraphics[width=.49\textwidth]{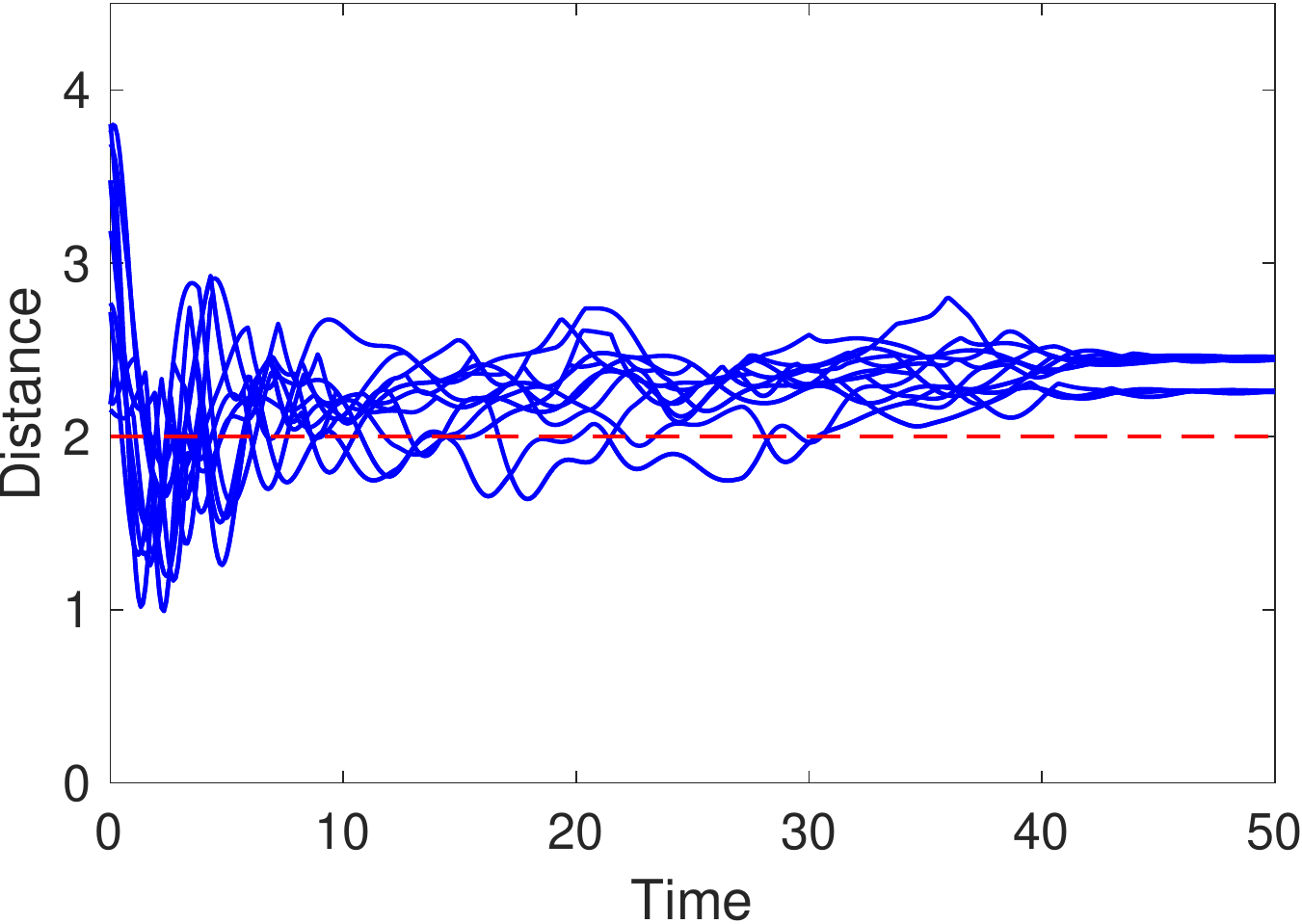}}\hfill
\subfloat[Reynolds Model with DSA]{\includegraphics[width=.49\textwidth]{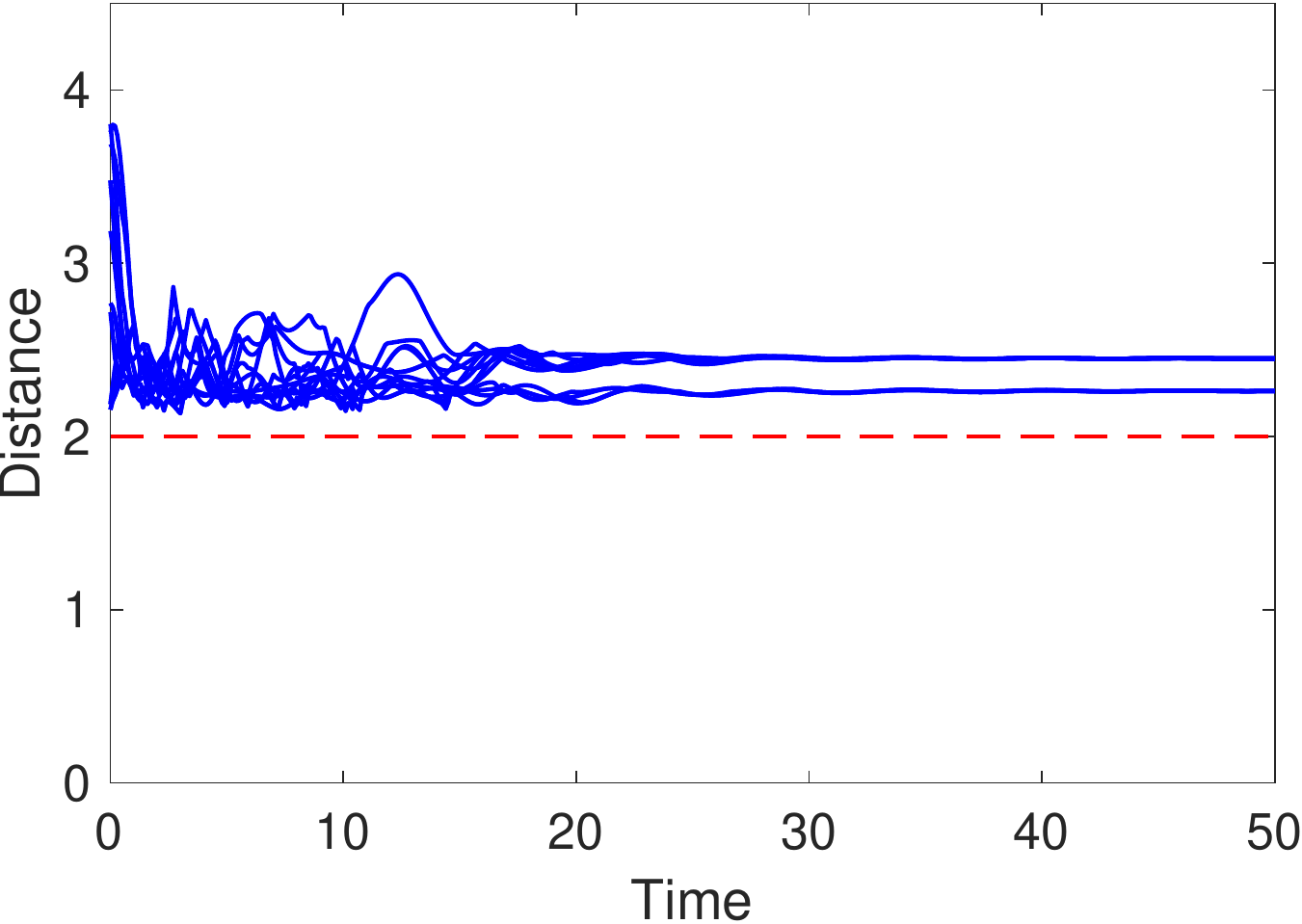}}\hfill
\caption{Results for a flock of size~15, with and without DSA.}\label{fig:plots}
\vspace*{-4ex}
\end{figure}

We performed two simulations, starting from the same initial configuration.
In the first simulation, Reynolds model is used to control all agents for the duration of the simulation. 
In the second simulation, Reynolds model is wrapped with a verified safe BC and DM designed using DSA. 
Videos of both simulations are available online.\footnote{\url{https://streamable.com/zn2bl5}\\ \url{https://streamable.com/hetraw}}

To recall, the safety property is that all pairs of agents maintain a distance greater than $d_{min}$ from each other.
Fig.~\ref{fig:plots} plots, for the duration of the simulations, the distance between the agents and their closest neighbors, 
i.e., ${\scriptstyle \underset{j \in \mathcal{N}_i}{min} \left \| p_i - p_j \right \|}$. 
As evident from Fig.~\ref{fig:plots}(a), Reynolds model results in multiple safety violations before the flock stabilizes at around 40 seconds. In contrast, as shown in Fig.~\ref{fig:plots}(b), DSA preserves safety, maintaining a separation greater than $d_{min}$ between all agents. DSA has an additional benefit of stabilizing the flock much earlier (around 18 seconds).
We further note that the average time the agents spent in BC mode is only $2.47$ percent of the total duration of the simulation, indicating that DSA is largely non-invasive.

\section{Way-point Control Case Study}
\label{sec:wp_case_study}

\begin{figure}[t]
\centering
\subfloat[The trajectories of the agents passing through the way-points. Red/blue segments indicate AC/BC modes.]{\includegraphics[width=.49\textwidth, trim=1.5cm 1.5cm 1.5cm 1.5cm]{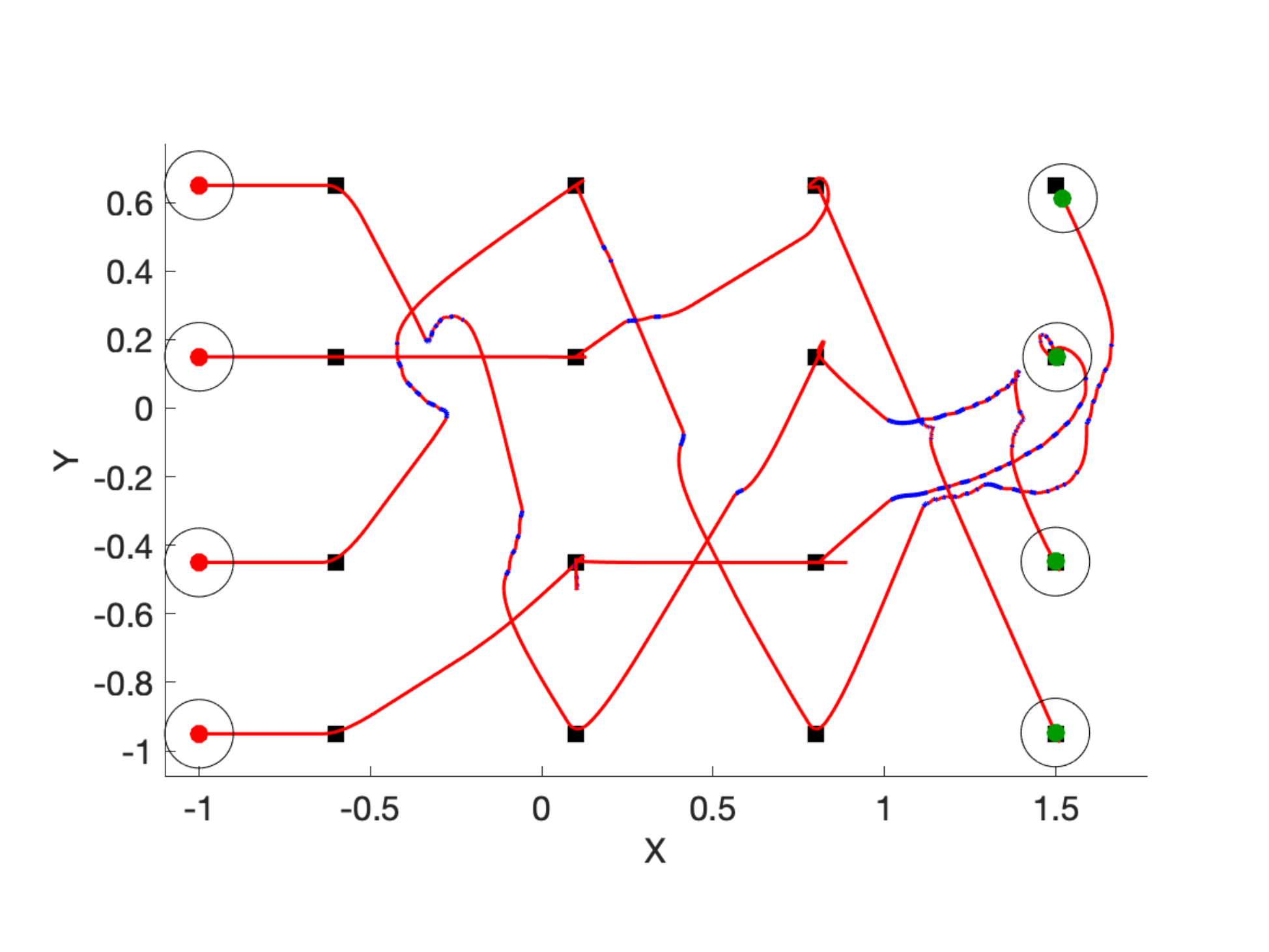}}\hfill
\subfloat[Distance to the closest neighbor for all agents.]{\includegraphics[width=.49\textwidth, trim=0cm -1.5cm 0cm 0cm]{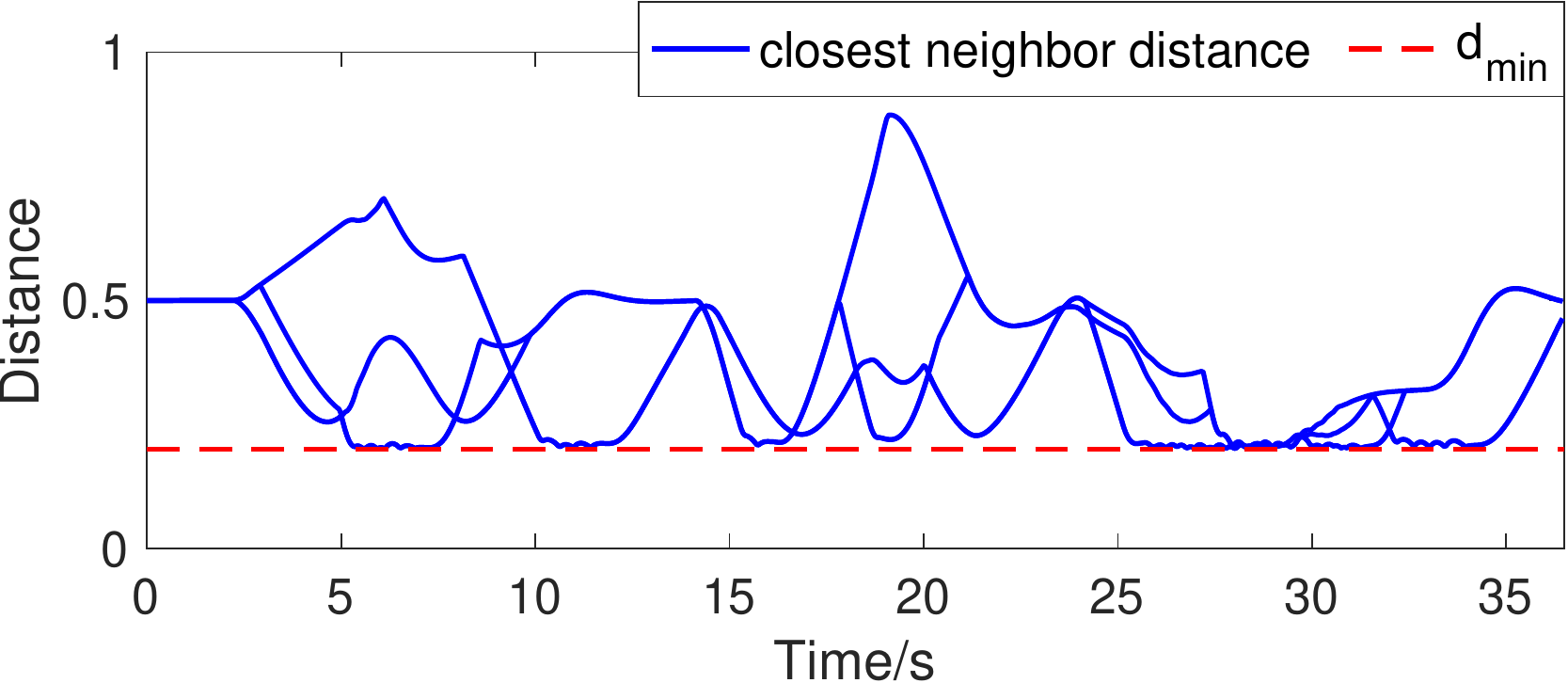}}\hfill
\caption{Experimental results for the way-point control study.}\label{fig:wp}
\vspace*{-4ex}
\end{figure}

This section describes the problem setup and experimental results for the way-point (WP) control case study.
For this study, the model of the agents is the same as the one used for the flocking case study, given in Eq.~(\ref{eq:flocking_dynamics_continuous}).
The experimental setup is shown in Fig.~\ref{fig:wp}, where the agents initially positioned on the left-hand side are to sequentially navigate through a series of WPs while maintaining a safe distance from each other. The WPs are represented by the black squares. 

The CBF, BC and DM are same as those defined for the flocking problem; see Section~\ref{sec:flocking_case_study}. The AC is a rule-based controller where each agent accelerates towards its next WP (ignoring the other agents) until the final WP is reached. 
Agents are assigned one WP from each column such that they are on a collision course if they follow the AC's commands.

\vspace{-2.0ex}
\subsection{Experimental Results}
The number of agents used in the experiment is four as is the number of WPs an agent is required to visit. Initially, the agents are at rest with their positions represented by the red dots in Fig.~\ref{fig:wp}(a). The final configuration is shown in green. 
The duration of the simulation is 37 seconds.
The other parameters used in the experiments are $r = 1.0$, $\bar{a} = 0.8$, $\bar{v} = 0.2$, $d_{min} = 0.2$, and $\eta = 0.05$s.
The trajectories of the agents are given in Fig.~\ref{fig:wp}(a), where the segments in blue indicate when the BC is in control. Fig.~\ref{fig:wp}(b) plots the smallest inter-agent distances, indicating that the agents maintain a safe distance from one another. A video of the simulation is available online.\footnote{\url{https://streamable.com/e9rnqd}}

\section{Microgrid Case Study}
\label{sec:mg_case_study}

With an increasing prevalence of distributed energy resources (DERs) such as wind and solar power, electrification using microgrids (MGs) has witnessed unprecedented growth. 
Unlike traditional power systems, MG DERs do not have rotating components such as turbines. The lack of rotating components can lead to low inertia, making MGs susceptible to oscillations resulting from transient disturbances~\cite{pogaku}.  Ensuring the safe operation of an MG is thus a challenging problem.
In this case study, we demonstrate the effectiveness of DSA in maintaining MG voltage levels within safe limits.

The MG we consider is a network of $n$ droop-controlled inverters, indexed by ${\mathcal{M} = \{1,\ldots,n\}}$. The dynamics of each inverter is modeled as~\cite{pogaku,droopeqns1, droopeqns2, Kundu2019}:
\begin{subequations}
\label{eq:droop_mg}
    \begin{equation}
       \dot{\theta_i}=\omega_i
    \end{equation}
    \begin{equation}
        \dot{\omega_i}=\omega_i^0-\omega_i+\lambda^p_i(P_i^{set}-P_i)
    \end{equation}
    \begin{equation}
        \dot{v_i}=v_i^0-v_i+\lambda^q_i(Q_i^{set}-Q_i)
    \end{equation}
\end{subequations}
where $\theta_i, \omega_i$, and $v_i$ are respectively the phase angle, frequency, and voltage of inverter $i$, $i \in \mathcal{M}$. 
The state vector for the MG is denoted by
$\mathbf{s}$ $=$ $[\theta^T \ \omega^T \ v^T]^T \in \mathbb{R}^{3n}$,
where $\theta$, $\omega$, and $v$ are respectively vectors representing the voltage phase angle, frequency, and voltage at each node of the MG. A pair of inverters are considered neighbors if they are connected by a transmission line. Also, $\lambda^p_i$ and $\lambda^q_i$ are droop coefficients of ``active power vs frequency'' and ``reactive power vs voltage'' droop controllers, respectively. Finally, $\omega_i^0$ and $v_i^0$ are the nominal frequency and voltage values.

$P_i$ and $Q_i$ are the active and reactive powers injected by inverter $i$ into the system:  
\begin{equation}
    \begin{aligned}
    \label{eq:PQ_eqns}
    P_i &= v_i \sum_{k \in \mathcal N_i}{v_k(G_{i,k} \cos{\theta_{i,k}}+B_{i,k} \sin{\theta_{i,k}})}\\
     Q_i &= v_i \sum_{k \in \mathcal N_i}{v_k(G_{i,k} \sin{\theta_{i,k}}-B_{i,k} \cos{\theta_{i,k}})}
     \end{aligned}
\end{equation}
where $\theta_{i,k} = \theta_i - \theta_k$, and $\mathcal N_i \subseteq \mathcal{M}$ is the set of neighbors. $G_{i,k}, B_{i,k}$ are respectively conductance and susceptance values of the transmission line connecting inverters $i$ and $k$. 

$P_i^{set}$ and $Q_i^{set}$ are the active power and reactive power setpoints. 
The inverters have the ability to change their respective power setpoints according to the MG's operating conditions. This is modeled as:
\begin{equation}
    \label{eq:PQ_setpoints}
    P_i^{set} = P_i^0+u_i^p, \
    Q_i^{set} = Q_i^0+u_i^q 
\end{equation}
where $P_i^0$ and $Q_i^0$ are the setpoints for the nominal operating condition, and $u_i^p$ and $u_i^q$ are control inputs. 

\subsection{Synthesis of Control Barrier Function}
The safety property for the MG network is a set of unary constraints restricting the voltages at each node to remain within safe limits. 
The recoverable set $\mathcal{R}_{i} \subset \mathbb{R}^3$ for inverter $i$ is defined as the super-level set of a CBF $h_i:\mathbb{R}^3 \to \mathbb{R}$. We follow the SOS-optimization technique given in~\cite{Kundu2019} to synthesize the CBF.

Since the power flow equations (\ref{eq:droop_mg}) are {nonlinear}, we apply a third-order Taylor series expansion to approximate the dynamics in polynomial form.  We then follow the three-step process given in~\cite{Kundu2019} to obtain the CBF for each MG node. 
We then calculate the admissible control space according to~(\ref{eq:admissible_control_space}),
and the BC, FSC, and RSC follow from (\ref{eq:bc}), (\ref{eq:FSC_DSA}), and (\ref{eq:RSC_DSA}), respectively.


\vspace{-2.0ex}
\subsection{Advanced Controller}
The AC sets the active/reactive power setpoints to their nominal values. 
Thus, the AC does not limit voltage and frequency magnitudes but is only concerned with stabilizing frequency and voltage magnitudes to their nominal values.

\vspace{-2.0ex}
\subsection{Experimental Results}
We consider a 6-bus 
MG~\cite{Kundu2019}. Disconnecting the MG from the main utility, we replace bus~0 with a droop-controlled inverter (Eq.~(\ref{eq:droop_mg})), with inverters also placed on buses 1, 4 and~5.
Bus~0 is the reference bus for the phase angle. Nominal values of voltage and frequency, as well as the active/reactive power set-points, were obtained by solving the steady-state power-flow equations given in Eq.~(\ref{eq:PQ_eqns}); these were then used to shift the equilibrium point to the origin.
Droop coefficients $\lambda_i^p$ and $\lambda_i^q$ were set to 2.43 rad/s/p.u. and 0.20 p.u./p.u., and $\tau_i$ was set to 0.5~s.
Loads are modeled as constant power loads, and a Kron-reduced network~\cite{kundur1994power} with only the inverter nodes was used for analysis. 
The unsafe set is defined in terms of the shifted (around the 0~p.u.) nodal voltage magnitudes as follows: $v_i < -0.4 \text{\,p.u. \,\,\text{or}\,\, } v_i > 0.2 \text{\,p.u.}$

\begin{figure}[t]
\centering
\includegraphics[width=.95\textwidth]{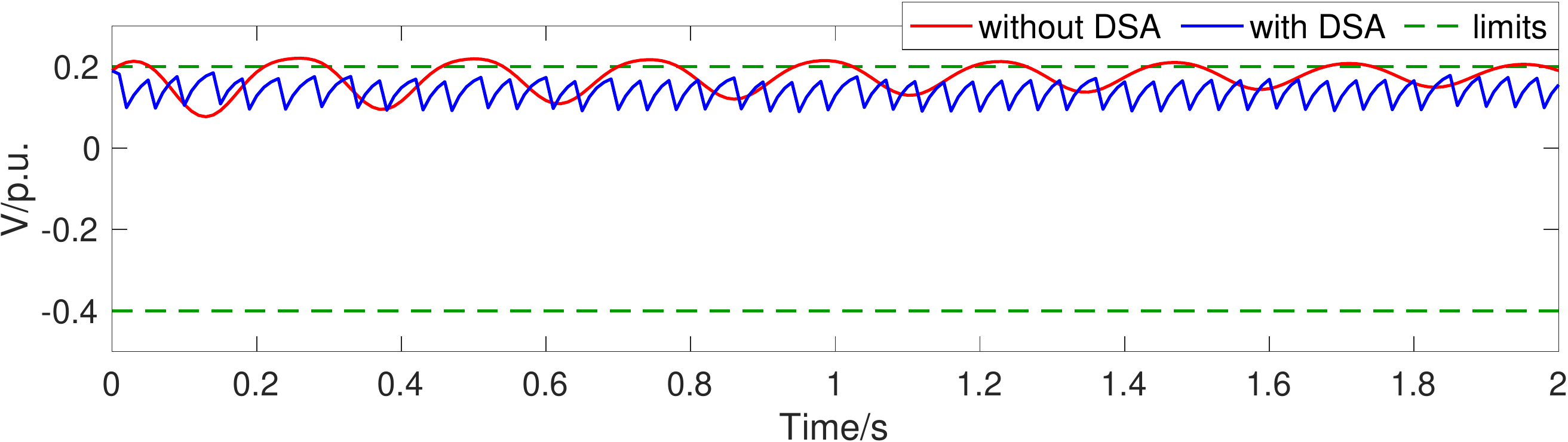}\hfill
\caption{Voltage graph at node~4 of the MG network.}\label{fig:mg_plots}
\vspace*{-4ex}
\label{fig:mg_res}
\end{figure}

The duration of the simulation is two seconds.  Our results show that with DSA, the voltages at all the nodes remain within the safe limits throughout the simulation. Fig.~\ref{fig:mg_res} gives the voltage plot at node~4. 
When the MG is operating under the control of DSA, forward and reverse switching cause sudden changes in the voltage.
As the voltage approaches the upper limit, a switch from AC to BC occurs. Subsequently, the BC reduces the voltage inducing a reverse switch.  
The voltage profile at the other nodes is similar. 

\section{Related Work}
\label{sec:related_work}
The original Simplex architecture~\cite{seto1998,seto1999} was developed for a systems comprising a single controller and a single (non-distributed) plant. With DSA, we extend the scope of Simplex to MASs under distributed control. 
%
RTA~\cite{Aiello2010,Schierman2012} is a runtime assurance technique 
that can be applied to component-based systems. In this case, however, each RTA wrapper (i.e., each Simplex-like instance) independently ensures a local safety property of a component. For example, in~\cite{Aiello2010}, RTA instances for an inner-loop controller and a guidance system are uncoordinated and  operate independently. 
In contrast, in DSA, each agent takes the states of neighboring agents into account when making control decisions, in order to ensure that pairwise safety constraints are satisfied.

A runtime verification framework for dynamically adaptive multi-agent systems (DAMS-RV) is proposed in~\cite{Lim2016}.  DAMS-RV is activated every time the system adapts to a change in the system itself or its environment. It takes a feedback loop- and model-based approach to verifying dynamic agent collaboration. 
However this method relies on a \emph{monitoring} phase  to observe and identify changes that occur in agent collaboration so that verification can be carried out on the system operating in new contexts. DSA does not require such intermediary supervision.
In~\cite{Alotaibi2010}, a dynamic policy model that can be used to express constraints on agent behavior is presented. These constraints limit agent autonomy to lie within well-defined boundaries. Constraint specifications are kept simple by allowing the policy designer to decompose a specification into components and define the overall policy as a composition of these small units.
In contrast, DSA uses CBFs to compute the requisite safety regions.

CBF-based methodologies~\cite{magnus2015,magnus_heterogenous2016,Ames2017,Ames2019_reviewCBF} have been widely used for MAS runtime safety assurance. In~\cite{magnus2015,magnus_heterogenous2016}, a formal framework for collision avoidance in multi-robot systems is presented.  
CBFs are used to design a wrapper around an AC that guarantees forward invariance of a safe set.  The wrapper solves an optimization problem involving the Lie derivative of the CBF to compute minimal changes to the advanced controller's output needed to ensure safety.

\section{Conclusion}
\label{sec:conclusion}
We have presented Distributed Simplex Architecture, a runtime assurance technique for the safety of multi-agent systems. DSA is distributed in the sense that it involves one local instance of traditional Simplex per agent such that the conjunction of their respective safety properties yields the desired safety property for the entire MAS. Moreover, an agent's switching logic depends only on its own state and that of neighboring agents.  We demonstrated the effectiveness of DSA by successfully applying it to flocking, way-point visiting, and microgrid control. As future work, we plan to apply DSA to  non-homogenous MASs 
and implement it on a physical platform.

\bibliographystyle{IEEEtran}
\bibliography{references}
\end{document}